\documentclass[preprint,12pt]{elsarticle}
\usepackage{epsfig}
\usepackage{latexsym}
\usepackage{amsmath}
\usepackage{amssymb}
\usepackage{amsthm}
\journal{Mathematical Biosciences}

\newtheorem{theorem}{Theorem}[section]
\newtheorem{proposition}[theorem]{Proposition}

\newtheorem{lemma}[theorem]{Lemma}
\newtheorem{conjecture}[theorem]{Conjecture}

\def\S{\mathcal{S}}

\newcommand{\RR}{{\mathbb R}}
 \newcommand{\PP}{\mbox{$\mathbb P$}}
\newcommand{\EE}{\mbox{$\mathbb E$}}

\begin{document}
\begin{frontmatter}
\title{Inferring ancestral sequences in taxon-rich phylogenies}

\author{Olivier Gascuel$^1$}
\author{Mike Steel$^2$}
\address{$^1$M{\'e}thodes et Algorithmes pour la Bioinformatique,
LIRMM, CNRS, Montpellier Universit{\'e} 161 rue Ada,  34392  Montpellier, FRANCE
Tel: +33 (0) 4 67 41 85 47; Fax: +33 (0) 4 67 41 85 00; Email:  gascuel@lirmm.fr}

\address{$^2$Corresponding author\\Allan Wilson Centre for Molecular Ecology and Evolution
University of Canterbury, Christchurch, NEW ZEALAND
Tel:  +64 3 3667001;   Fax: +64-3-3642587; Email: m.steel@math.canterbury.ac.nz}

\begin{abstract}
Statistical consistency in phylogenetics has traditionally referred to the accuracy of estimating phylogenetic parameters for a fixed number of species as we increase the number of characters. However, as sequences are often of fixed length (e.g. for a gene) although we are often able to  sample more taxa, it is useful to consider  a dual type of statistical consistency where we increase the number of species, rather than characters.  This raises some basic questions: what can we learn about the evolutionary process as we increase the number of species?  In particular, does having more species allow us to infer the ancestral state of characters accurately?  This question is particularly relevant when sequence site evolution varies in a complex way from character to character, as well as for reconstructing ancestral sequences. In this  paper, we assemble a collection of results to analyse various approaches for inferring ancestral information with increasing accuracy as the number of taxa increases.
\end{abstract}
\date{Mar. 24, 2010}

\end{frontmatter}

\section{Introduction}

As Elliott Sober discussed two decades ago \cite{sob}, there is a fundamental asymmetry between reconstructing a past state from a present observation, and predicting its future state. Moreover, this holds even when the state evolves according to a time-reversible process (processes which, when they are in equilibrium, behave the same whether run forward or backward in time).  For instance, consider any continuous Markov process on two states, with arbitrary transition rates (generally unequal) between the two states.  If we observe the state of the process at the present time $t$, then the `best' estimate of the initial state at time $0$ is always the present state, but the `best' estimate of its state at some future time $t'>t$ depends on the actual transition rates (which may be unknown) \cite{sob}.

When we move beyond two states in a Markov process, the current state is no longer guaranteed to always be the `best' estimate of the ancestral state,  even for reversible processes, as we describe below. Ancestral state estimation assumes a further dimension when we move from the linear evolution of a state through time to the bifurcating evolution of states in a tree that results in their observed values at the leaves.  The presence of many leaves helps us to estimate the ancestral state more accurately, but these leaves do not provide independent information
about the root state due to correlations arising from the partial overlap of the paths in the tree as one moves from the root to the leaves.
The mathematical, statistical and computational aspects of ancestral state estimation on a tree  have been explored by a number of authors (e.g. \cite{li, mad, mos, pag, sal, zha}) and the inference of ancestral states is an important question in biology \cite{lib}.

Our interest  here is in site-specific models. These are especially relevant with proteins, where each site has specific biochemical constraints (e.g. small and hydrophobic, aromatic, helix-former, etc). As we are interested in
site-specific models, the details of the substitution model are mostly unknown. For example, the relative or absolute branch length may not be known exactly, though we may have some upper bound on them.    Also, the equilibrium
frequencies at the site may not be known.  This is the case in the CAT model for proteins (\cite{lar}; see also \cite{kos}). This model is a mixture of F81-like models, where each site follows  a Poisson model with specific character frequencies defined by the biochemical constraints acting on that site. However, we shall see that dealing with unknown equilibrium frequencies imposes strong limitations when the aim is to estimate ancestral characters, especially when the branch lengths are unknown. Thus, we will also envisage special cases where equilibrium frequencies are known or even all identical.

In most cases (e.g. when the branch lengths are unknown), we are thus unable to use standard likelihood calculations based on the pruning algorithm to compute the most likely character at the tree root. Thus, we will discuss and study simple decision rules to predict the state at the tree root. Parsimony is an example of such a rule, where the branch lengths are useless. Another example is the majority rule that involves selecting the state that is most frequent at the tree leaves to estimate the root state.  For models in which the equilibrium frequencies are not uniform across states, more complex inference rules are required.
We shall see that under suitable assumptions on the tree topology and branch lengths and/or on the model, these simple rules are statistically consistent as the taxon sampling density becomes sufficiently large.

We treat four general cases, each depending on the properties of the model. We start with the simplest (symmetric) model, then consider two overlapping generalizations (`monotone' and `conservative') and finally we deal with the general model, for which stronger assumptions on the tree are required.

\subsection{Preliminaries}

Consider a rooted phylogenetic tree $T$ (possibly non-binary) with $n$ leaves and a set $\S$ of possible states that each vertex can be in.  For a single-site assignment of states at the leaves of $T$,  assume that the assignment has evolved under a GTR (general time-reversible) model from a particular character state $s_0$ at the root, with a normalized rate matrix $Q = \Pi S$ (where $\Pi = {\rm diag}(\pi)$ contains the equilibrium frequencies, and $S$ is a symmetric matrix of `exchangeabilities').  The process acts on each edge $e$ according to some associated branch length $l_e$.  A more general version of this question is when a single site is replaced by a (possibly short) sequence (of length $k$).
Assuming independent site evolution, the problem of ancestral state estimation remains the same (i.e. each site is solved independently).

We assume that $T$ and $S$ (and perhaps $\pi$) are given, and, in addition, we may either know $l_e$ or have some bounds on them (e.g. the sum of the lengths from the root to any tip is, at most, some given value of $t$). We would like to use this input to estimate the ancestral state $s_0 \in \S$ at the root of the tree.  The ability to estimate $s_0$ accurately depends on a tradeoff between what we know about the underlying parameters  (e.g. the site rate parameter $\mu$, the branch lengths $l_e$, and the properties of $Q$ such as the equilibrium distribution $\pi$) and how `well behaved' the underlying Markov process is.

In particular, we seek a method $M$ of estimation that is statistically consistent in the following sense:  As $n$ becomes large, and given increasingly tight constraints on the tree, its branch lengths or the model, we want the accuracy of $M$ (the probability that $M$ reconstructs the ancestral state correctly) to converge to 1.

A natural choice of such a method, when $Q$ is completely specified (including the equilibrium distribution $\pi$) and the branch lengths ($l_e$) are also known exactly, is to take the maximum posterior probability (MPP) ancestral state (this selects the state with the largest posterior probability;  the MPP method can be shown to confer the largest expected reconstruction probability  amongst all methods). Note that for a symmetric model the MPP estimate of the root state is the same as the maximum likelihood (ML) estimate, but in general the two approaches differ (because the prior distribution of the states at the root multiplies these ML values by the prior  in the MPP approach).

When the whole model is (partly) unknown, the ML and MPP approaches may no longer  be feasible.  But in these cases, simpler approaches exist. For example, for a simple symmetric model (e.g. Jukes-Cantor) and a  star tree with unknown branch lengths that are bounded above ($l_e \leq l < \infty$), we can estimate the ancestral state accurately by selecting the majority state (the consistency of this approach is justified by  large deviation theorems for sums of independent random variables).

However, even for symmetric models, it is clear that simply allowing $n$ to grow is not sufficient to allow for accurate inference of the ancestral state $s_0$; for example, we could have just two long edges incident with the root, and lots of very short edges that join the other endpoints of these edges to numerous taxa.  In this case, the tree basically behaves as a two-taxon tree and we have little information on the root when the two branches become too long.  Thus we seek relevant and reasonable constraints on the distribution of $l_e$ values for this accurate estimation to be possible.  One possibility  for generating taxon-dense trees is to evolve a Yule (pure birth) tree of total height $t$ and to select a large speciation rate $\lambda$ (we may then re-scale the rate on each edge by some bounded multiplicative factor to allow for violation of a strict molecular clock).

Moving away from symmetric models, selecting the majority state at the leaves as an estimate of the ancestral state  is not generally a sound strategy, even for a star tree, since the process after a long period of time will favour the state with the highest equilibrium frequency, regardless of the state at the root.

\section{Case I: Root state estimation without detailed knowledge of  $l_e$ under a symmetric Poisson model}
\label{parsec}
Under the symmetric $r$-state Poisson model, the maximum likelihood estimate of the root state, in the case where  the branch lengths ($l_e$) are unknown and are regarded as nuisance parameters to be optimized, is the maximum parsimony (MP) estimate (Theorem 6 of \cite{tuf}). In this setting, we can reliably estimate the root state, provided the taxon sampling is sufficiently dense that no edges are too long. This was suggested by the simulations in \cite{sal} and we establish two formal results now for the case when $r=2$.

\begin{proposition}
\label{nice3}
Consider any rooted binary phylogenetic tree $T$. Evolve a single site under the two-state symmetric model. Let $l_+$ be the maximum
branch length over all edges.  Provided that $l_{+} < \frac{1}{2}\log(\frac{4}{3})$, the probability $P^*$  that the maximum parsimony (MP) reconstruction of the root state is the true state (toss a fair coin if the two states are equally favored) satisfies:
$$P^* \geq  1-3l_+.$$
\end{proposition}
\begin{proof}
When $l_{+}$ satisfies satisfies the bound described then, for each edge $e$ of $T$ the probability that the endpoints of edge $e$ are in different states  $p(e)=\frac{1}{2}(1-e^{-2l_e})$ satisfies the inequality  $p(e) < \frac{1}{8}$.
It then follows from  part (ii) of Lemma 5.1 of \cite{tease}, that:
$$P^* \geq \frac{1}{2}+ \Delta_g,$$ where: $$\Delta_g = \frac{\sqrt{(1-4g)(1-8g)}}{2(1-2g)^2}, $$
and where $g = \max_e \{p(e)\}$.  The result now follows from the inequalities:
$$ \frac{\sqrt{(1-4g)(1-8g)}}{2(1-2g)^2} \geq \frac{1}{2}(1-6g), \mbox{ and } g \leq l_+. $$
\end{proof}

Unfortunately, in a Yule tree of fixed height, the expected value of $l_+$ does not converge to zero as the speciation rate $\lambda$ tends to infinity. This may seem surprising, since each edge in the tree converges in length to $0$ as $\lambda$ grows; however, the expected number of edges increases with $\lambda$, and the probability that at least one of them is `long' turns out to be positive.
 Simulations suggest that the expectated value of $l_{+}$ converges to a value
close to 60\% of the height of the tree; the following result, the proof of which is provided in the Appendix, establishes a smaller lower bound.

\begin{proposition}
\label{yuleprop}
Suppose a random rooted binary tree $T_\lambda$ is generated by a Yule (pure birth) process with speciation rate $\lambda$ acting for time $t$.  Let $l_{+} = l_{+}(\lambda)$ denote the length of the longest edge in $T$.
Then $\EE[l_{+}(\lambda)]$ does not converge to $0$ as $\lambda \rightarrow \infty$.
\end{proposition}

Thus we cannot directly
apply Proposition \ref{nice3} to Yule trees.   Nevertheless, we can precisely determine the probability with which MP will correctly reconstruct the root state of a Yule tree under a symmetric Poisson substitution model on two states. In particular, provided the speciation rate passes a critical threshold (six times the substitution rate), then even for large trees where many leaves are far from the root, ancestral reconstruction is feasible.  Moreover, as the ratio of speciation rate to substitution rate tends to infinity,  we can correctly infer the root state with probability tending to $1$.

\subsection{MP root estimation for Yule trees under the two-state model}

Consider a pure-birth Yule tree that starts with a single (root) lineage at time $0$ and is grown until time $t$, with speciation rate $\lambda$.  
Suppose we  also have a binary character that evolves from some ancestral state at the root of the tree towards the leaves by  undergoing substitution along the edges of the tree  at rate $\mu$ according to a symmetric Markov process.
  Thus, we have a random tree (with a random number of leaves at time $t$) and a random binary character observed at the leaves. 
   Let $P_t$ denote the probability that the maximum parsimony estimate  for the state at the root of the tree, derived from the
observed states at the leaves of the tree at time $t$, matches the true root state (in the case that both states are equally parsimonious at the root, select one state with equal probability).
\begin{theorem}
\label{yulepars}
\begin{itemize}
\item[(i)]
If $\lambda  \geq 6 \mu$, then for all $t \geq 0$:
$$P_t \geq \frac{1}{2}(1 + \sqrt{(1-6\rho)(1-2\rho)}) \geq 1-3\rho;$$
where $\rho = \mu / \lambda$.
In particular, $P_t \rightarrow 1$ as $\rho \rightarrow 0.$
Moreover, $P_t$ is monotone decreasing in $t$ with  limit:
$$\lim_{t \rightarrow \infty} P_t= \frac{1}{2}(1 + \sqrt{(1-6\rho)(1-2\rho)}).$$
\item[(ii)]
If $\lambda  < 6 \mu$ we have:
$$\lim_{t \rightarrow \infty} P_t= \frac{1}{2}.$$
\end{itemize}
\end{theorem}

\begin{figure}[ht]
\includegraphics[width=0.8 \textwidth] {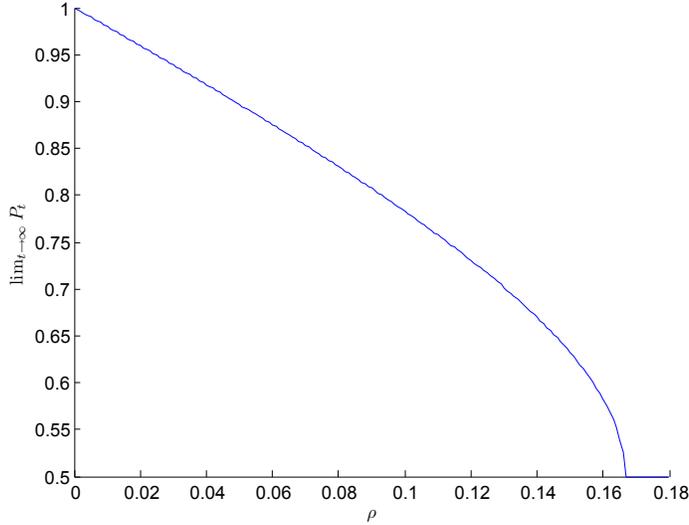}
\caption{The limiting value $\lim_{t\rightarrow \infty} P_t$ as a function of $\rho$ for $\rho \leq 1/6$. }
\label{fig:mat}
\end{figure}

\begin{proof}
Let $0$ and $1$ denote the two states that undergo substitution on the Yule tree. Since the Markov process is symmetric we may suppose, without loss of generality, that $0$ is the initial character state at time $t=0$.  From the (random) evolved states on the leaves, estimate the root state using the maximum parsimony criterion (i.e. select the root state that minimizes the total number of substitutions required to describe the evolution of the character on the tree).   There may be a unique reconstructed root state (which may be the same or opposite to the true initial state) or both states may be equally parsimonious.  
Let $S_t$ (resp. $D_t$) be the probability that $0$ (resp. $1$) is the unique most parsimonious root state reconstructed from the observed states at the leaves.   Let $E_t = 1-S_t-D_t$ be the probability that both states are equally parsimonious.
We have:
\begin{equation}
\label{right}
P_t = S_t + \frac{1}{2}E_t = \frac{1}{2} + \frac{1}{2}(S_t-D_t).
\end{equation}
We can generate a system of non-linear first-order differential equations for $(S_t, D_t, E_t)$  as follows.  Consider that in the first $\delta$ period of time,
the root lineage can either:
 \begin{itemize}
 \item  persist, without a substitution occurring,
  \item
  persist, with a substitution occurring, or
  \item  it can speciate into two lineages.
  \end{itemize}
This gives:
$$S_{t+\delta} = (1-\mu\delta - \lambda \delta) S_t + \mu \delta D_t + \lambda \delta (S_t^2 + 2S_tE_t) + O(\delta^2).$$
Similarly,
$$D_{t+\delta} = (1-\mu\delta - \lambda \delta) D_t + \mu \delta S_t + \lambda \delta (D_t^2 + 2D_tE_t) + O(\delta^2),$$
$$E_{t+\delta} = (1-\mu\delta - \lambda \delta) E_t + \mu \delta E_t + \lambda \delta (E_t^2 + 2S_tD_t) + O(\delta^2).$$
Rearranging these expressions and letting $\delta \rightarrow 0$ produces the differential equation system:
$$\frac{dS_t}{dt} + (\lambda + \mu) S_t = \mu D_t +\lambda(S_t^2 + 2S_tE_t);$$
$$\frac{dD_t}{dt}  + (\lambda + \mu) D_t = \mu S_t + \lambda(D_t^2 + 2D_tE_t);$$
and
$$\frac{dE_t}{dt} +  (\lambda + \mu) E_t = \mu E_t + \lambda(E_t^2 + 2S_tD_t).$$

Notice that we can use the relationship $S_t+D_t+E_t =1$ to eliminate $E_t$, and by writing $u = \lambda t$ we obtain a
two-dimensional autonomous differential equation system for $S= S_u, D=D_u$:
$$\frac{dS}{du} = f(S, D); \mbox{ }  \frac{dD}{du}=f(D, S),$$
where:
$$f(x,y) = (1-\rho)x+\rho y -2xy - x^2.$$
Now, $(S,D)$  is confined to the simply-connected, two-dimensional, compact region $S, D \geq 0, S +D  \leq 1$, and
we can analyse its dynamics using standard phase-portrait methods for autonomous two-dimensional dynamical systems (see e.g. 
\cite{str}).  We note first, that $(S,D)$ has no limit cycle by virtue of Dulac's criterion (with $g(x,y) = 1/xy$, for details see \cite{str}).   From the starting condition $(P,D) = (1,0)$, at $u=t=0$, the quantity $\Delta_u = S_u - D_u$ is non-negative and monotone decreasing, and $(S,D)$ converges to
an asymptotically stable equilibrium point, which can be found by solving the system $\frac{dS}{du}=\frac{dD}{du}=0$ and carrying out
an eigenvalue analysis  of the Jacobian of the system.

Solving  $\frac{dS}{du}=\frac{dD}{du}=0$ is equivalent to solving the pair of simultaneous quadratic equations $f(s,d) =0, f(d,s)=0$. Subtracting the second of these equations from the first gives:
\begin{equation}
\label{eq1}
(s-d)(1-2\rho -s-d)=0.
\end{equation}
Thus, either $s=d$ or $s+d= 1-2\rho$.
If $s=d$, then the equation $f(s,d)=0$ becomes $s-3s^2=0$, which has two possible solutions: either
$s=d=\frac{1}{3}$ or $s=d=0$;  the first of these is asymptotically stable when $\rho>1/6$.

In the other case, where $s+d=1- 2 \rho$, $f(s,d)=0$ becomes:
$$s^2-(1-2\rho)s + \rho(1-2\rho)=0,$$
which also has two possible solutions:
$$s = \frac{1-2\rho \pm \sqrt{(1-6\rho)(1-2\rho)}}{2}$$
both of which are asymptotically stable with $\rho < 1/6$.
Since $s-d\geq 0$ (since $\Delta_u>0$), the positive sign applies in the previous equation.
Since in this case $e=1-s-d=2\rho$, we have: $$s+\frac{1}{2}e =  \frac{1}{2}(1 + \sqrt{(1-6\rho)(1-2\rho)}).$$

The results stated in the theorem now follow, since
 Eqn.~(\ref{right}) allows us to write:
\begin{equation}
\label{right2}
P_{\lambda t} = \frac{1}{2}+\frac{1}{2}\Delta_u
\end{equation}
and so $P_t$
is monotone decreasing with $t$, and we also have  the inequality
 $\frac{1}{2}(1 + \sqrt{(1-6\rho)(1-2\rho)}) \geq 1-3\rho$  when $\rho<1/6$.
\end{proof}

It would be interesting to obtain corresponding results for maximum parsimony for more general models - particularly for symmetric models on more than two states (some limited results are described in \cite{ste}, Sections 9.4.1 and 9.5.1). Here we offer the following:

\begin{conjecture}
For the $r$-state symmetric model Proposition~\ref{nice3} generalizes to give a  lower bound on $P^*$ of $1-c_r \cdot l_+$ for some constant $c_r>0$.  Similarly, Theorem~\ref{yulepars} generalizes to give an analogous result, where the critical ratio
$\lambda/\mu = 6$ is replaced by $\lambda/\mu = c'_r$ for some constant $c'_r>0$.
\end{conjecture}

\section{Case II:   Conservative GTR proceses}
\label{consec}

 For any Markov process, we often write $\PP_i(X_t = j)$ or $p_{ij}(t)$ for the conditional probability $\PP(X_t=j|X_0=i)$ that $X_t =j$ given that $X_0 = i$.
We will say that a GTR model is {\em conservative}  if, for every state $i$ we have: $$p_{ii}(t) > p_{ij}(t) \mbox{ for all } t \geq 0 \mbox{ and all } j \neq i.$$
This is the `forward inequality' described by Sober  \cite{sob}.
The Kimura two-parameter (K2P) model (and every submodel, such as Jukes-Cantor) is an example of a conservative process (see Fig.~\ref{fig:mat2}).
In this model the substitution probabilities are given as follow (for details see \cite{swo}):
$$p_{ij}(t) = \begin{cases}
\frac{1}{4}(1+e^{-\mu t} + 2e^{-\mu t(\kappa+1)/2}),  &\text{if~} i=j;\\\
\frac{1}{4}(1+e^{-\mu t} - 2e^{-\mu t(\kappa+1)/2}), &\text{if~} i \rightarrow j \mbox{ is a transition; } \\\
\frac{1}{4}(1-e^{-\mu t}), &\text{if } i \rightarrow j \mbox{ is a transversion.}\
\end{cases}
$$

\begin{figure}[ht]
\includegraphics[width=0.8\textwidth] {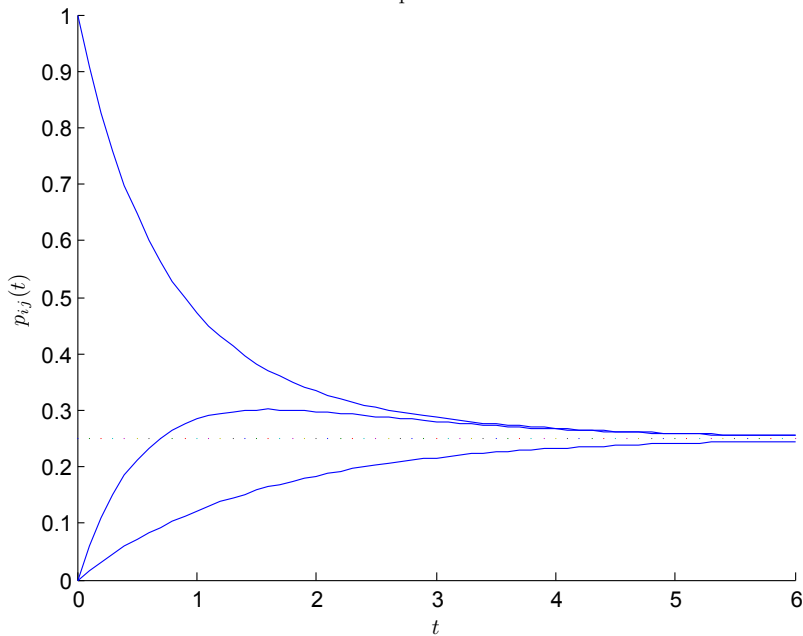}
\caption{The three substitution probabilities for the Kimura 2ST model. This model is conservative (but not monotone). In  this example,  $\kappa =4$ and $\mu$ is chosen to be $2/3$ so that $t$ corresponds to the expected
number of substitutions. The curve descending from $1$ is the function $p_{ii}(t)$. The middle curve, which has a local maximum around $1.6$ is the probability of a transition (A $\leftrightarrow$ G or C $\leftrightarrow$ T ); the lower curve is the probability of a transversion (Purine (A or G) $\leftrightarrow$ Pyrimidine (C or T).}
\label{fig:mat2}
\end{figure}

With a conservative model, the majority rule applies for ancestral reconstruction.  Assuming state $i$ at the tree root, the probability of observing $i$ at any tree leaf is higher than the probability of observing any
 particular alternative state $j$. This holds true for whatever the root-to-leaf distances and the tree topology. With a star tree, with an upper bound on the root-to-leaf distances, the probability that this inference rule makes the correct selection tends to $1$ as the number of leaves grow (by the central limit theorem for sums of independent random variables).  We shall see that this result still holds  for a more general class of trees under mild assumptions.   We now describe this class of trees and their properties.

\subsection{Well-spread trees}

Given a rooted phylogenetic $X$-tree and a leaf $x \in X$, let:
$$l_x := \sum_{e \in P(\rho, x)}l_e,$$ the sum of the branch lengths on the path $P(\rho, x)$ from
the root of the tree ($\rho$) to leaf $x$, and where $X$ is the set of $n$ leaves.
Similarly, for distinct leaves $x,y \in X$, let:
$$l_{xy} =  \sum_{e \in P(\rho, x)\cap P(\rho, y)}l_e,$$
the total length of the shared paths from $\rho$ to the leaves $x,y$.  Finally, define the {\em spread} of $T$ as:   $$s(T) := \frac{\sum_{x,y} \min\{l_{xy}, 1\}}{n(n-1)}.$$
Thus, provided $l_{xy}<1$ for all $x,y$,  $s(T)$ the average value of $l_{xy}$ over pairs $x,y$.
We say that $T$ is {\em well spread} if $s(T)$ is small; more precisely, $T$ is $1-\beta$ spread if $s(T) \leq \beta$. In particular, a
tree is a star tree if and only if it is 1-spread.

Note that a well-spread tree must have a large number of edges close to its root; an example is shown in Fig.~\ref{fig:trees}.
\begin{figure}[ht]
\includegraphics[width=1.0 \textwidth] {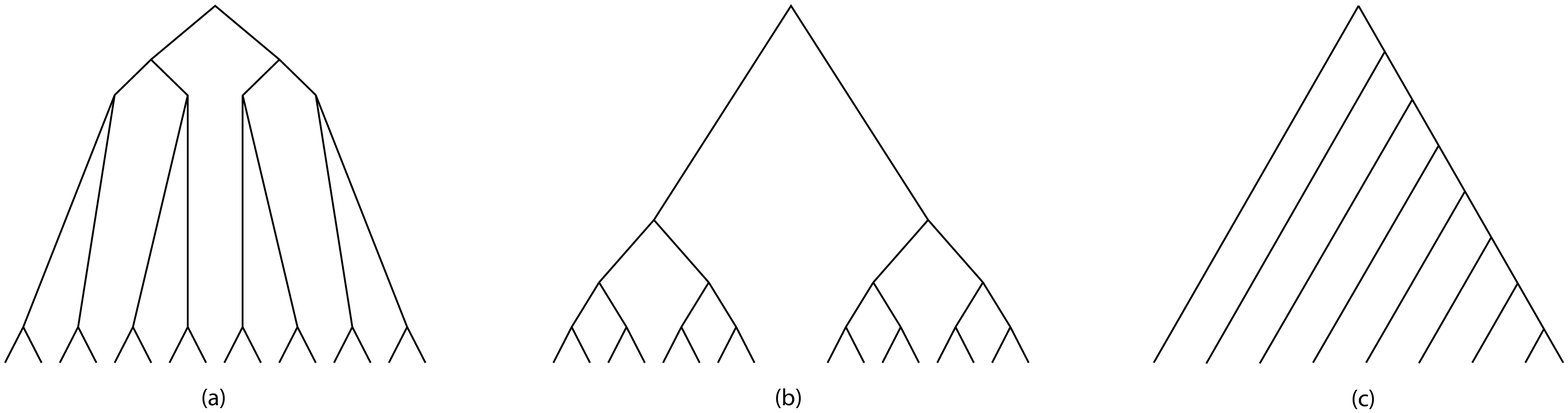}
\caption{(a) A well-spread tree; (b, c) Trees that are not well-spread.}
\label{fig:trees}
\end{figure}

It is easily shown that a sufficient condition for a tree to be $1-\beta$ spread is that, for some $\epsilon, \delta>0$ with $\epsilon+\delta< \beta$,
the proportion of pairs of leaves whose paths from the root to those leaves overlap by a length of at least $\epsilon$ is no more than $\delta$.
We use this observation to show that the spread of a Yule pure-birth tree of fixed height $t$ approaches $1$ as the speciation rate grows.

\begin{proposition}
Consider a random Yule pure-birth tree $T$  that has speciation rate $\lambda$ and fixed height $t$. Then for any $\beta>0$, the probability that $T$ is $1-\beta$ spread converges to $1$ as $\lambda \rightarrow \infty$.
\end{proposition}
\begin{proof}
We may assume that $T$ has a root of out-degree $2$ (the length of a single lineages shrinks to zero with probability $1$ as $\lambda$ grows).   By Theorem 2(2) of \cite{mck}, the expected proportion of pairs of leaves
whose most common ancestor lies $r$ or more edges from the root of the tree has the geometric probability $(2/3)^r$.  Given $\epsilon, \delta>0$  with $\epsilon+\delta <\beta$ first select a  sufficiently large value of $r$ that
$(2/3)^r \leq \delta$.  For any $\eta>0$ we can now select a  sufficient large value of $\lambda$ that the probability that all the (at most) $2^r$  vertices separated from the root by $r$ edges have are within distance $\epsilon$ from the root is
at least $1-\eta$. The result now follows.
\end{proof}

We now introduce some further notation.  For each state $j \in \S$, let $n_j$ denote the number of leaves of $T$ that are in state $j$, and let   $\rho_j = \rho^i_j$ be the expected proportion of leaves that are in state
$j$, given that the root is in state $i$.  Thus $n_j$ is a random variable (whose distribution depends on the root state $i$)  while $\rho_j$ is a value determined by the model parameters, $j$ and root state $i$.

The following Lemma is central to many of the results that follow in this paper.

\begin{lemma}
\label{centrallem}
Suppose that $T$ is a rooted tree, with branch lengths, and which  is $1-\beta$ spread.
  Then for any continuous-time Markov process on $T$,   the following holds for all initial states $i$:
For any $s>0$, the probability of the event that for all states $j \in \S$:
$$\left|\frac{n_j}{n}-\rho^i_j\right| <s$$
is at least $1- f(n, \beta)/s^2$ where $f(n,\beta)$ tends to $0$ as  $\max\{\frac{1}{n}, \beta\} \rightarrow 0$.
\end{lemma}
\begin{proof}
 For $x \in X=\{1, \ldots, n\}$, let $\theta_x^j$ be the random variable that takes the value $1$ if leaf $x$ is in state $j$ and $0$ otherwise.  We have
 $\frac{n_j}{n} = \frac{1}{n}\sum_{i=1}^n \theta_x^j$ and
 $\rho_j = \frac{1}{n}\sum_{i=1}^np_{ij}(l_x).$
In particular, since $p_{ij}(l_x) = \EE[\theta_x^j]$, linearity of expectation gives:
$$\EE\left[\frac{n_j}{n}\right] = \rho_j.$$
Now:
\begin{equation}
{\rm Var}\left[\frac{n_j}{n}\right] = n^{-2}\left (\sum_{x \in X} {\rm Var}[\theta_x] + \sum_{x,y \in X, x \neq y}{\rm Cov}[\theta_x, \theta_y] \right),
\label{vareq}
\end{equation}
and ${\rm Var}[\theta_x]  \leq \frac{1}{4}$,  $|{\rm Cov}(\theta_x, \theta_y)| \leq 1$.
Moreover, for any pair $(x,y)$ we claim that  $|{\rm Cov}(\theta_x, \theta_y)| \leq K\min\{1, l_{xy}\}$ for a constant $K$ dependent only on the model.  To see this, let $N_{xy}$ be the event that the root ancestral state does not change state anywhere along the shared path of length $l_{xy}$. We have $\PP(N_{xy})= \exp(-cl_{xy})$ for a constant $c$ dependent only on the model. Moreover, the random variables
$\theta_x, \theta_y$ are conditionally independent, given $N_{xy}$.  Routine algebra then shows that we can express
${\rm Cov}[\theta_x, \theta_y]$ as  $l_{xy}$ times a constant, plus terms of order $l_{xy}^2$. However, since in addition
${\rm Cov}[\theta_x, \theta_y]\leq 1$,  we have ${\rm Cov}[\theta_x, \theta_y]\leq \min\{Kl_{xy}, 1\} \leq K \min\{l_{xy}, 1\}$ for some sufficiently large constant $K >1$.
Thus, from \ref{vareq}, we have:
$${\rm Var}\left[\frac{n_j}{n}\right] \leq  \left(\frac{1}{4n} + Ks(T) \right).$$
Let $f_1(n, \beta) = (\frac{1}{4n}+ K\beta)$ then,  since $s(T)<\beta$, and by Chebyshev's inequality, we have:
$$\PP\left(|\frac{n_j}{n}-\rho_j| \geq s \right) \leq \frac{{\rm Var}\left[\frac{n_j}{n}\right] }{s^2} \leq f_1(n, \beta)/s^2.$$
Thus, if $r = |\S|$ denotes the number of possible states, and if we let $f(n, \beta) :=  rf_1(n, \beta)$ then
we have:
$$\PP\left(\exists j:  |\frac{n_j}{n}-\rho_j| \geq s \right)  \leq f(n, \beta)/s^2,$$
which converges to zero when both $n \rightarrow \infty$ and
$\beta \rightarrow 0$.
\end{proof}

\begin{theorem}
For a conservative model and a $1-\beta$ spread tree, with $l_x \leq l$ for each $x$, the probability that the majority state at the leaves is identical to the ancestral state at the root is at least $1-g(n,\beta, l)$ for a function $g$ which (for each value $l$)  tends to zero as $\max\{\frac{1}{n}, \beta\} \rightarrow 0$.
\end{theorem}
\begin{proof}
Let $$\delta_l :=\min_{i,j: i \neq j} \inf\{p_{ii}(t)-p_{ij}(t): t \in [0,l]\}.$$
Since $p$ is continuous, and $[0,l]$ is compact, the conservative property implies that $\delta_l>0$.
Moreover,  we have:
\begin{equation}
\label{rhoav}
\rho^i_i - \rho^i_j \geq \delta_l \mbox{ for all } j \neq i.
\end{equation}
Now by Lemma~\ref{centrallem},  the probability of the event  that $\left|\frac{n_j}{n}-\rho^i_j\right| <\frac{1}{2}\delta_l$ for all $j$ is
at least $1-4f(n,\beta)/\delta_l^2$.  Moreover, for this event, Inequality (\ref{rhoav}) implies (by the triangle inequality) that
$\frac{n_i}{n} - \frac{n_j}{n} >0$ for all $j \neq i$; that is, the ancestral state $i$ is the majority state at the leaves.
Thus the probability that the ancestral state is the majority state is at least
$1-g(l, n, \beta)$ where $g(l,n,\beta) := 4f(n, \beta)/\delta_l^2$ has the required stated properties.
\end{proof}

\section{Case III:   Monotone time-reversible proceses}
\label{monsec}

Note that, for {\em any} general time-reversible (GTR) Markov process the function $p_{ii}(t)$  is always
monotone decreasing to its equilibrium frequency $\pi_i$ for each state $i$  \cite{ald}, that is:
$$p_{ii}(t) > p_{ii}(t') \mbox{ for all } t< t'.$$

We will say the model is  {\em monotone} if, for all distinct states $i,j$, we have:
 $$p_{ij}(t) < p_{ij}(t') \mbox{ for all } t< t'.$$  Thus a  monotone model has the property that
 if we start in a particular state $i$ then the probability that we are in a different particular state $j$ at time $t$ increases monotonically with $t$ towards  its equilibrium probability $\pi_j$.
 In particular, a monotone GTR model satisfies the `backward inequality'  from \cite{sob} that $p_{ii}(t) > p_{ji}(t)$ for all $j \neq i$ and $t \geq 0$, since $p_{ji}(t)$ is monotone increasing to $\pi_i$ while
 $p_{ii}(t)$ is montone descreasing to $\pi_i$.

For any number of states, models such as the Felsenstein 1981 model (also called the  F81, Tajima-Nei, or Equal Input model)  are monotone (but not conservative, unless all equilibrium frequencies are equal).  Also any  two--state Markov process is monotone (Fig.~\ref{fig:two-state}) and the  implications of this for biological inference on the basis of a single observation ($n=1$) were explored in \cite{sob} and \cite{sob2}.

 \begin{figure}[ht]
\includegraphics[width=0.6\textwidth] {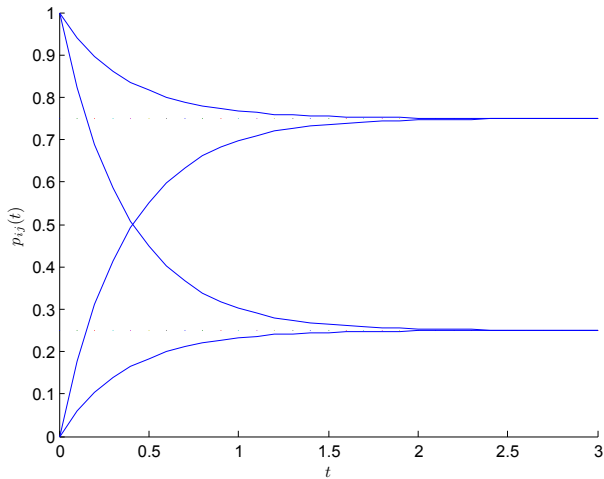}
\caption{A two-state model with different equilibrium frequencies ($0.75, 0.25$) for the states $0$ and $1$, respectively. The two decreasing curves are $p_{00}(t)$ (upper) and $p_{11}(t)$ (lower). The two increasing curves are $p_{10}(t)$ (upper) and $p_{01}(t)$ (lower).  This model is monotone, but not conservative.  }
\label{fig:two-state}
\end{figure}

Amongst nucleotide substitution models, the  K2P model is not monotone, since  if $i \neq j$
represents a transition then: $$p_{ij}(t) = \frac{1}{4}+\frac{1}{4}e^{-\mu t} - \frac{1}{2}e^{- \mu t(\frac{\kappa+1}{2})}$$ can behave as shown by the middle curve in Fig.~\ref{fig:mat2}, where $\kappa$ is the transition-transversion ratio (taken to be a
default option of $4$ here).

Despite K2P not being monotone, this model nevertheless satisfies Sober's `backward inequality' as it is a symmetric model (i.e. $p_{ij}(t) = p_{ji}(t)$ for all $t$); however more complex time-reversible continuous Markov processes can fail this inequality. For example, consider a process on states $0,1,2,\ldots, m$ with equal and high transition rates from each value of  $k$ (less than $m$)  to $k+1$ and equal  low transition rates from each $k$ (greater than $1$) to $k-1$. Then for a suitably large value of $m$ and choice of $t=t_1$, we have $p_{11}(t_1) < p_{01}(t_1)$.  In particular, observing state $1$ at a particular (known) time $t_1$ provides more evidence that the initial state was $0$ rather than $1$.

With monotone models, the majority rule can be misleading -- when the time $t$ is larger than the time corresponding to the intersection point  $p_{ii}(t)=p_{ij}(t)$,  it is  more likely to have  $j \neq i$
 at any given leaf than to have $i$. However, simple prediction rules still exist, which depend on what is known/unknown.

When the equilibrium frequencies are known, we use the fact that if $i$ is the ancestral state then the proportion of taxa in state $i$, $\frac{n_i}{n}$ is expected to be larger than  $\pi_i$  (at least if the number of taxa is sufficient to avoid sampling effects), while $\frac{n_j}{n}$ is expected to be less than $\pi_j$ for all $j \neq i$. This suggests a {\em modified majority rule}: select as an ancestral state estimate the state  $i$ which maximizes  $\frac{n_i}{n} - \pi_i$.  Note that the branch lengths and even the tree topology do not need to be known. Moreover, this decision rule becomes the simple majority rule when the equilibrium frequencies are all equal.

However, with site specific models, we cannot assume that the equilibrium frequencies are known, especially with proteins (as discussed earlier). In such a case, we can use a second decision rule, based on the fact that $p_{ii}(t)$   is a decreasing function of $t$  while $p_{ij}(t)$  is increasing. This second decision rule needs the root-to-leaf distances to be known and variable across taxa (however the tree topology may be unknown).  Let $\overline{l_j}$  be the average distance between the root and the taxa having state $j$, and let   $\overline{l_{-j}}$  be the average distance between the root and the taxa having a state different to $j$.  As the distance of a leaf from the root increases, the probability that leaf is in the ancestral state $i$ should also decrease, while the reverse trend should hold for any other state $j$. In other words, we select $i$  to minimize  $\overline{l_i}-\overline{l_{-i}}$.  Note that for this rule to apply, we need the root-to-taxon distance to be sufficiently heterogeneous. With a molecular clock-tree this rule is of no help. Moreover, we do not need to know the site rate and the absolute branch lengths, and the topology and branch lengths may be unknown, provided we still can estimate the root-to-leaf distances.

We shall see that under mild assumptions, both rules for monotone models are statistically consistent. We now describe the two procedures for monotone models more precisely, depending on whether $\pi$ is known or not. We then state a theorem that provides conditions under which
these estimators are accurate.

The two procedures are as follows:

\begin{itemize}
\item   {\bf $\pi$ known:}    Select the ancestral state $i$ to maximize $\frac{n_i}{n} - \pi_i$.
\item   {\bf $\pi$ not known:}  Select the ancestral state $i$ to minimize $\overline{l_i}- \overline{l_{-i}}$.
\end{itemize}

We will show that the first estimator performs well provided the tree is well spread and  $n$ is large.
 The second estimator requires, in addition, that there be reasonable spread amongst the  root-to leaf distances (i.e. that they be {\em not} clocklike).  First, we require a lemma which is
 a mild extension of Chebyshev's order inequality (the proof is given in the Appendix).

 \begin{lemma}
 \label{covlem}
 Suppose that $Y$ is a random variable taking values in $[0,l]$ and that  $f: [0,l] \rightarrow \RR$ is a smooth function with $f'(y) \geq  c>0$  for all $y \in [0, l]$. Then:
 $${\rm Cov}[Y, f(Y)]  \geq c \cdot {\rm Var}[Y].$$
 Similarly, if $f'(y) \leq -c < 0$  for all $y \in [0, l]$ then  ${\rm Cov}[Y, f(Y)]  \leq -c \cdot {\rm Var}[Y].$
 \end{lemma}

\begin{theorem}
\label{ans}
Suppose we have a monotone GTR model and $\alpha>0$.
\begin{enumerate}
\item
The first estimation procedure described above (for a known $\pi$) correctly selects the true ancestral state with probability at least $1-\alpha$ provided the following three conditions hold:
\begin{itemize}
\item[(i)] $l_x \leq l < \infty$, for all $x$, and some $l$ independent of $n$;
\item[(ii)]  $T$ is $1-\beta$ spread for  sufficiently small values of $\beta$,  and
\item[(iii)] $n$ is sufficiently large.
\end{itemize}
\item
The second estimation procedure described above (for $\pi$ not known) correctly selects the true ancestral state with probability at least $1-\alpha$ provided that, in addition to conditions (i) -- (iii),  the  following two conditions hold:
\begin{itemize}
\item[(iv)] The variance of the $l_x$ values is greater or equal to some fixed value $v>0$ as $n$ grows.
\item[(v)] $\pi_j \in (0,1)$ for all $j \in \S$.
\end{itemize}
\end{enumerate}
\end{theorem}
\begin{proof}
For part (1), let $\delta_1 = \min_{i,j: i \neq j} \{\pi_j - p_{ij}(l)\}, \delta_2=\min_i \{p_{ii}(l)-\pi_i\}$ and $\delta_l = \min\{\delta_1, \delta_2\}$.
By the monotonicity property, we have $\delta_l >0$.
If $i$ is the ancestral state then:
\begin{equation}
\label{rhoav2}
 \rho^i_i \geq \pi_i + \delta_l,  \mbox{ and for any state $j \neq i$, }  \rho^i_j \leq \pi_j - \delta_l.
 \end{equation}
Now by Lemma~\ref{centrallem}, the probability of the event  that $\left|\frac{n_j}{n}-\rho^i_j\right| <\delta_l$ for all $j$ is
at least $1-f(n,\beta)/\delta_l^2$.  Moreover, for this event, Inequality (\ref{rhoav2}) implies (by the triangle inequality) that
$\frac{n_i}{n} -\pi_i >0$ and for all $j \neq i$, we have $\frac{n_j}{n} -\pi_j<0$, in which case the correct ancestral state ($i$) will be selected
by the decision rule.   Thus if we select a sufficiently small value of $\beta$ and a sufficiently large value of $n$ that $1-f(n, \beta)/\delta_l^2 < \alpha$ we obtain the result in Part (1).

For part (2), we show that if $i$ is the ancestral state then, with high probability, $\overline{l_i} - \overline{l_{-i}}<0$ and for all $j \neq i$, $\overline{l_j} - \overline{l_{-j}} >0$.
For any state $j$ (including $i$),  consider the difference:
$$D_j: = \overline{l_j} - \overline{l_{-j}}.$$
Recalling the definition of $\theta_x^j$ from the proof of Lemma~\ref{centrallem} we have:
$$\overline{l_j} = \frac{\sum_{x \in X}l_x\theta_x^j}{n_j} \mbox{ and } \overline{l_{-j}} = \frac{\sum_{x \in X}l_x(1-\theta_x^j)}{(n-n_j)},$$ and so:
\begin{equation}
\label{difference}
D_j = \frac{\frac{1}{n}\sum_{x\in X} l_x\theta_x^j - (\frac{1}{n}\sum_{x \in X}l_x) \cdot \frac{n_j}{n}}{\frac{n_j}{n}(1-\frac{n_j}{n})}.
\end{equation}
By assumption (v), $D_j$  is well defined (i.e. $n>n_j>0$ in the denominator) with probability converging to $1$ as $n$ grows.
Let $$\overline{l}: = \frac{1}{n}\sum_{x \in X} l_x, \mbox{ and let } \overline{L} := \frac{1}{n}\sum_{x\in X} l_xp_{ij}(l_x).$$
Notice that we can write the numerator of $D_j$ in the form:
\begin{equation}
\label{formeq}
(\overline{L} - \overline{l}\rho^i_j ) + \left (\frac{1}{n}\sum_{x\in X} l_x \theta_x^j - \overline{L}\right) + \overline{l}(\rho^i_j - \frac{n_j}{n}).
\end{equation}

Now, let $c_1 = \min_{j \neq i} \inf \{\frac{dp_{ij}(t)}{dt}: t \in [0,l]\}$ and $c_2 =  \inf \{\frac{-dp_{ii}(t)}{dt}: t \in [0,l]\}$, and
$c = \min\{c_1, c_2\}$.  By the monotone assumption, $c>0$.  We can now apply Lemma~\ref{covlem} as follows.  Define a random variable $Y$ by setting $Y = l_x$ for a leaf $x$ selected uniformly at random from  the leaf set $X$, and
let $f(y) = p_{ij}(y)$.  Then, ${\rm Cov}[Y, f(Y)] = \overline{L} - \overline{l}\rho^i_j$ and so, by Lemma~\ref{covlem}, we have:
\begin{equation}
\label{meanineq}
\overline{L} - \overline{l}\rho^i_i \leq -cv, \mbox{ and } \overline{L} - \overline{l}\rho^i_j \geq c v\mbox{ for all $j \neq i$},
\end{equation}
where $v>0$ is a lower bound on the variance of the $l_x$ values from condition (iv).
Note that $\EE[\frac{1}{n}\sum_{x\in X} l_x \theta_x^j] = \overline{L}$, and  since $l_x \leq l$ for all $x \in X$, an argument similar to that given in Lemma~\ref{centrallem} implies that $\left|\frac{1}{n}\sum_{x\in X} l_x \theta_x^j - \overline{L}\right|$ can be made less than any $\delta>0$ by selecting $\beta$ and $\frac{1}{n}$ sufficiently small. Moreover, the same applies for the difference $\left | \frac{n_j}{n} - \rho^i_j \right|$ by Lemma~\ref{centrallem}.
 Thus, from expression (\ref{formeq}), the numerator of $D_j$ can be made arbitrarily close to the difference $\overline{L} - \overline{l}\rho^i_j$ by selecting  $\beta$ and $\frac{1}{n}$ sufficiently small.
It then follows from  Inequality (\ref{meanineq}) that  the sign of $D_j$ will be negative for $j=i$ and positive otherwise, as required (noting that $c$ depends just on the model, not on $\beta$ or $n$).  This completes the proof.

\end{proof}

\section{Case IV:  Non-monotone and non-conservative models}

Some simple and widely used models are neither monotone nor conservative. For example, the `HKY' (Hasegawa, Kishino and Yano) model combines both K2P and F81 (\cite{swo}); as with K2P, the transition probabilities first increase and then decrease (non-monotony); because the equilibrium frequencies may be unequal, the probability of observing the root state $i$ at a leaf may be less than the probability of observing state $j$ when  $\pi_i < \pi_j$.

 \begin{figure}[ht]
\includegraphics[width=0.8\textwidth] {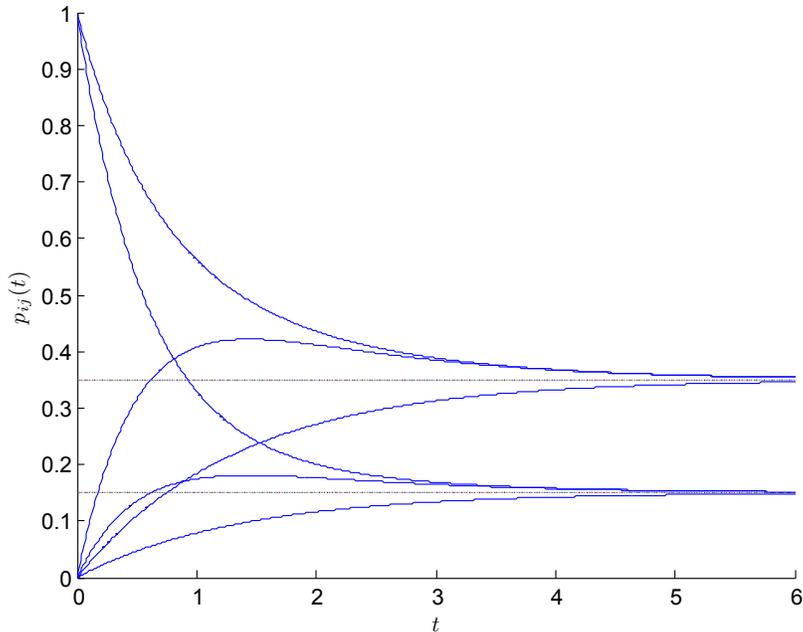}
\caption{HKY transition probabilities with standard parameter values ($\kappa = 4$,  purine=pyrimidine=0.5, GC = 70\%).  The asymptotic values $0.35$ and $0.15$ are the equilibrium GC and AT frequencies, respectively. The two decreasing curves are $p_{ii}(t)$ (e.g. G $\rightarrow$ G and C$\rightarrow$ C for the top-most curve). The two increasing curves with local maxima are for transitions (e.g. A $\rightarrow$ G, T $\rightarrow$ C for the top-most increasing curve) while the two monotone increasing curves are for transversions.}
\label{fig:mat3b}
\end{figure}

With such models, the justifications provided for the statistical consistency of the  three simple rules above and parsimony no longer apply. However, when the model is fully known and the tree is clock-like, the ancestral state can still be estimated using the frequencies of the characters at the tree leaves. We  shall see that this method is statistically consistent.

We first state a general result concerning general Markov processes.
\begin{lemma}
\label{ident}
Consider any continuous-time, irreducible Markov process, and let $X_t$ be the state at time $t$. Then
for any given $t \geq 0$, the probability distribution on $X_t$ determines both $X_0$ and $t$. That is:
 $$\PP_i(X_t =j) = \PP_{i'}(X_{t'} = j) \mbox{ for all } j \in \S \Rightarrow i=i', t=t'.$$
\end{lemma}
\begin{proof}
Let ${\bf e}_{i}$ be the vector that has  $1$ in position $i$ and $0$ otherwise, and define ${\bf e}_{i'}$ analogously. Now, the vector
$p^{i}(t):=[\PP_i(X_t =j): j \in \S]$ satisfies $p^{i}(t) = {\bf e}_{i}\exp(Qt)$; similarly we have
$p^{i'}(t') = {\bf e}_{i'}\exp(Qt')$.  Suppose values of $t, t'$ exist for which $p^{i}(t) = p^{i'}(t')$. Without loss of generality, we may suppose that $t\geq t'$.  In this case we have:
$$({\bf e}_{i}\exp(Q(t-t')) - {\bf e}_{i'})\exp(Qt') = 0.$$
Moreover, since the process is irreducible, ${\bf e}_{i}\exp(Q(t-t'))$ can equal ${\bf e}_{i'}$ only if  $t=t'$ and $i = i'$, so if this is not the case, we have  ${\bf w}\exp(Qt') = 0$ for a non-zero vector ${\bf w}$ which implies that
 $$\det\exp(Qt') = 0.$$ But, by Jacobi's identity, $\det\exp(Qt') = \exp(tr(Q)t') >0$.  This completes the proof.
\end{proof}

From this Lemma, it follows that for the very special case of a star tree with all edges of equal length we can use maximum likelihood to consistently infer the ancestral state.  This is because, in this very special case,  the states at the $n$ leaves provide $n$ i.i.d. samples of the process, and so the identifiability conditions required to estimate $s_0$ and $\mu$ hold (for similar reasons to
the tailored argument for the consistency of MLE in settings such as phylogenetic tree reconstruction, described in Lemma 5.1 of \cite{cha}).
Moving from star trees to the more general class of well-spread trees, we have the following main result of this section:

\begin{theorem}
\label{nice}
Suppose we have a  continuous-time irreducible Markov process with rate matrix $Q$ given, and let $\alpha>0$.
Consider a rooted phylogenetic tree on $n$ leaves, for which the branch lengths $l_e$ satisfy a molecular clock, i.e. $l_x = l_0$ for all $x$, where $l_0$ is  less than some
known value $l$. Assume also that the tree is $1-\beta$ spread.
Then we can  estimate the ancestral state $s_0$ correctly with probability at least $1-\alpha$ provided that $n$ is sufficiently large, and $\beta$ is sufficiently small.
\end{theorem}
\begin{proof}
We will establish this result by a procedure that selects the state $i$  for which the entire probability distribution $p_{ij}(t)$ (as $j$ varies) can be made the `closest' to the empirical distribution $\frac{n_j}{n}$ for an optimal value of $t$. We will use the $l_\infty$ metric to measure `closeness' (although, in applications other metrics may be preferable) so we will select the ancestral estimate $i$ if $i$ minimizes the
quantity: $$\inf_{t \in [0,l]} \max_{j} \left|\frac{n_j}{n}-p_{ij}(t)\right|.$$
First observe that, for any two states $i,' \in \S$ with $i \neq i'$, if we let:
$$\Delta_{ii'}: = \inf_{t,t' \in [0,l]} \max_{j \in \S} \{|p_{ij}(t)-p_{i'j}(t')| \},$$ then $\Delta_{ii'}>0$ by Lemma \ref{ident}, the compactness of $[0,l]$ and the continuity of $p$.  Thus
$\delta_l := \min_{i,i': i \neq i'} \Delta_{i,i'}$  is also strictly greater than zero.  Notice that $\delta_l$ is independent of $n$.
Suppose that $i$ is the true ancestral state and $i'$ is a different state.  By the molecular clock assumption, $\rho^i_j = p_{ij}(l_0)$. Thus,
by Lemma~\ref{centrallem}, the probability of the event  that $\left|\frac{n_j}{n}-p_{ij}(l_0)\right| <\frac{1}{2}\delta_l$ for all $j$ is
at least $1-4f(n,\beta)/\delta_l^2$.  Moreover, for this event:
$$\inf_{t \in [0,l]} \max_{j} \left|\frac{n_j}{n}-p_{ij}(t)\right| < \inf_{t' \in [0,l]}\max_{j} \left|\frac{n_j}{n}-p_{i'j}(t') \right|$$
since the left-hand side is less than $\frac{1}{2}\delta_l$ and if $t'$ is the value that minimizes the right-hand side then, by the triangle inequality for the $l_\infty$ metric:
$$\max_{j} \left|\frac{n_j}{n}-p_{i'j}(t') \right| \geq \max_j \left|p_{ij}(t)-p_{i'j}(t') \right| - \max_j\left|\frac{n_j}{n}-p_{ij}(t) \right| \geq \Delta_{ii'}- \frac{1}{2}\delta_l \geq \frac{1}{2}\delta_l.$$
Thus, the selection method will choose the correct ancestral state ($i$) with probability at least $1-4f(n,\beta)/\delta_l^2$ and, as before, this can be larger than $1-\alpha$ by ensuring
that $\frac{1}{n}$ and $\beta$ are sufficiently small.
\end{proof}

When the tree is non-clock like,  and the model and branch lengths are known,  we might use a standard ML approach based on the pruning algorithm, though the precise conditions required for statistical consistency seem less clear.

\section{Simulations}

To compare the convergence rate and the performance of the various ancestral character reconstruction methods discussed in the previous sections, we performed computer simulations under biologically realistic conditions similar to  \cite{des}.

We first generated a Yule tree with $n =$  25, 50, 100, 200, 400, 800 and 1600 leaves. This molecular-clock tree was then perturbed by multiplying every branch length (independently) by $(1 + X)$, where $X$ was an exponential variable with parameter 0.5. The factor $(1 + X)$ was used (as opposed to, say, $X$) to avoid an excessive number of very small branches. The observed departure from the molecular clock, as measured by the ratio between the longest and shortest root-to-leaf lineages, was equal to $\approx3.5$ on average, a value that is usual in published phylogenies. Finally, the whole tree was re-scaled so that the average root-to-leaf distance was uniformly distributed between 0.1 (relatively low divergence) and 1.0 (high divergence).

DNA-like sequences of 100 sites were evolved along this tree using the HKY model with $\kappa=4.0$ (default value in most software) and the equilibrium frequencies of A, C, G and T being equal to 0.15, 0.35, 0.35 and 0.15, respectively (such GC bias is observed in thermophilic bacteria and archaea, while {\em Plasmodium} species have an even stronger AT bias). The same parameter values are used in Fig. \ref{fig:mat3b}. This HKY model was combined with a discrete gamma distribution of parameter 1.0 with six rate categories. We generated 500 data sets under these settings for each tree size $n$.

Five ancestral character prediction methods were compared:
\begin{itemize}
\item
`Parsimony'  (studied in Section~\ref{parsec});
\item
`Majority' (studied in Section~\ref{consec});
\item
`Modified majority', when the equilibrium frequencies $\pi$ are known (studied in Section~\ref{monsec}, {\em cf.}  part (1) of Theorem~\ref{ans});
\item
`Difference of average root-to-leaf distances',  when the equilibrium frequencies $\pi$ are unknown, but we know the root-to-leaf distances (studied in Section 4, {\em cf.} part (2) of Theorem~\ref{ans});
\item
`Presence',  which involves drawing with equal probability one of the characters that are present at the tree leaves. Indeed, it frequently occurs (notably with small $n$) that not all four possible characters are observed at the tree leaves. Moreover, all previous prediction methods never output a character that is not seen at the tree leaves. This implies that the difficultly of the prediction problem depends on the number of extant character states, and thereby depends on $n$. In the extreme case where we observe a unique extant character, all methods achieve perfect predictions (unless hidden convergent substitutions), while when the four characters are observed the chance is $1/4$ to be correct by chance. `Presence' is thus used to re-scale the performance of the various methods, depending on $n$ and the hardness of the prediction problem.
\end{itemize}

All methods were run with perfect knowledge of the tree topology (`Parsimony'), equilibrium frequencies $\pi$ (`Modified majority') or root-to-leaf distances (`Difference of average root-to-leaf distances'). For each method and each data set, we measured:
\begin{itemize}
\item
The percentage of correct predictions;
\item
The rescaled percentage of correct predictions, using the results achieved by `Presence'.  Let $P$ be the percentage of correct predictions of the given method, and $R$ be the percentage of correct predictions of `Presence'; the rescaled percentage of correct predictions is equal to $(P-R)/(1-R)$ and measures the fraction of improvement brought by the given method compared to random predictions.
\end{itemize}

Results averaged over 500 data sets are reported in Table 1 for each tree size $n$. We see that:
\begin{itemize}
\item
The results of `Presence'  indicates that the hardness of the prediction problem increases when $n$ increases; with $n=25$ the number of extant characters is around two on average, while it is around four with $n=1600$, meaning that the problem is `twice as hard' with $n=1600$ as compared with $n=25$.

\item
The accuracy of all methods improves with large $n$. However, the rescaled percentage of correct predictions is required to see this effect with `Difference of average root-to-leaf distances',  which is the method with the slowest convergence rate.
\item
Surprisingly,  `Parsimony' is slightly behind `Majority' and `Modified majority'. This finding is also observed with JC69 symmetrical model (results not shown), and thus cannot be attributed to the chosen substitution model (HKY); it is likely due to the fact that some of the simulated trees show a high divergence, a condition where `Parsimony' tends to perform poorly (see Theorem 2.3).

\item
Both `Majority' and `Modified majority' are very close, while we expected the latter to be better because it makes use of the equilibrium frequencies $\pi$. The explanation is likely related to the fact that in our simulations the root-to-leaf distance is less than 1.0 in average, a condition where HKY is basically conservative ({\em cf.} Fig. \ref{fig:mat3b}) and thus `Majority'   is consistent. However, we see a small superiority of `Modified majority' with large $n$, when the estimations of the $n_i/n$ frequencies become sufficiently reliable.  Moreover, HKY is monotone up to $\approx1.45$ while it is conservative up to $\approx0.8$ only.

\item
Finally, the performance of `Difference of average root-to-leaf distances'  is rather low, but there is a clear improvement with large $n$. This confirms that root-to-leaf distances bring substantial information, which could be combined with other standard approaches to enhance accuracy in difficult cases.
\end{itemize}

\begin{table}[ht]
\centering
\caption{Average accuracy with simulated data.   For each method we provide the percentage of correct predictions and (within parentheses) the rescaled percentage of correct predictions (see text for definition). 500 data sets with 100 sites each were used for each number of taxa ($n$). Abbreviations are Mod. Majority: `Modified  majority',  Diff. Aver. Dist.: `Difference of average root-to-leaf distances'. }
\label{table2}
\begin{tabular}{lccccc}
$n$ & Parsimony & Marjority & Mod. Majority & Diff. Aver. Dist.  &   Presence   \\ \hline
 25 & 0.820 (0.652)	& 0.832 (0.674)	 & 0.824 (0.659)	& 0.609 (0.214)	 & 0.499	 \\ \hline
 50 & 0.841 (0.728) & 	0.852 (0.746) & 	0.846 (0.736)& 	0.570 (0.237)	&  0.433  \\ \hline
 100 & 0.853 (0.772)	&  0.863 (0.788) & 	0.860 (0.784)	&  0.539 (0.262) & 	0.371 \\  \hline
 200 & 0.864 (0.802) & 	0.870 (0.811) & 	0.871 (0.813)	&  0.521 (0.285) & 	0.326 \\  \hline
 400 & 0.873 (0.822)  & 	0.880 (0.833)	 & 0.886 (0.842)	&  0.522 (0.324) &   0.289  \\  \hline
  800 & 0.885 (0.844) & 	0.885 (0.844) & 	0.896 (0.858)	&  0.537 (0.362)  & 0.270 \\  \hline
  1600 & 0.890 (0.852)	 & 0.891 (0.853) & 	0.906 (0.873)	&  0.567 (0.410) &  0.261  \\  \hline
\end{tabular}
\end{table}

All together, the most surprising outcome of these simulations is the performance of the (very simple) `Majority'  approach. It must be emphasized that `Majority' does not use any additional knowledge (tree topology, root-to-leaf distances or equilibrium frequencies), meaning that the gap could be larger if the other methods (e.g. `Parsimony') were used with only approximate knowledge (e.g. tree topology).

\section{Discussion}

In this paper, we have described and analysed five approaches for inferring ancestral root state in taxon-rich trees: maximum parsimony, simple majority rule, modified majority rule, root-to-leaf differences, and best-fit of
expected distribution of leaf states to the empirical distribution.  The methods are all relatively simple and easily implemented, and require different model (and tree) assumptions in order to justify their accuracy.   They can be applied in settings where one does not have enough information to carry out a full maximum likelihood analysis using the usual pruning algorithm, and so may be more suitable for site-specific models, where the process of evolution is likely to vary in a partially unknown way from character to character.

The price one might expect to pay for a method that requires fewer
assumptions or detailed knowledge of underlying parameters is lower accuracy.   Nevertheless, we have described several results which show that these methods (particular to the type of model in question) can still return the correct
ancestral state provided that the number of taxa ($n$) is sufficiently large, and the tree is sufficiently well-spread.   We have shown that for Yule trees with a high speciation rate (as a token for high taxon coverage), we expect a tree of fixed height to become increasingly well-spread as $n$ grows.  It is clear that some type of assumption  on the spread of the tree is necessary to avoid having two long branches near the root and the majority of  lineage splitting
well away from the root, in which case accurate root state inference is not possible.

Except for maximum parsimony, the methods described do not use the tree topology explicitly (only the distribution of states at the leaves, and perhaps their distance from the root are employed)  and so may be more robust to tree mis-specification.  Of the class of models described monotone models are perhaps the most relevant for application, since most GTR models are likely to be monotone (and even conservative) when restricted to amounts of evolutionary change that are commonly encountered for sequence evolution.

Our choice of methods to study  in this paper has been guided by what can be usefully analysed, and we are not advocating these methods above others that might be considered; in particular, we make no claim that they are `best possible'. Indeed, if one has sufficient information then more standard approaches such as maximum likelihood would be preferable. However, the simplicity of these methods, and the  fact that they are relatively robust to model mis-specification may make them a useful complement to more sophisticated approaches. It is also possible to develop statistical tests to determine whether differences observed in the data by our approaches are significant or not.
For future studies, it would be worthwhile to explore the performance of these approaches on biological data-sets, comparing them with other alternative approaches that have been advocated; however.

\section{Acknowledgments}
We thank David Aldous for suggesting the example of a reversible Markov process that violates Sober's `backward inequality'.

\newpage

\section{Appendix: Proof of Proposition~\ref{yuleprop} and Lemma~\ref{covlem}}

For the proof of Proposition \ref{yuleprop}  let $n = e^{\lambda t/2}$. Then the expected number of taxa at time $t/2$ is $n$ and is $n^2$ at time $t$. Let $N_u$ be the number of individuals at time $u$.  Let $E_1$ be the event that
$N_{\frac{t}{2}}$ lies between $\frac{1}{2}n$ and $\frac{3}{2}n$,  let $E_2$ be the event that $N_t  < 2n^2$, and let $E$ be the conjunction of $E_1, E_2$. We first establish the following:

CLAIM:  For some $\delta>0$,  $\PP(E)\geq \delta$, for all sufficiently large $\lambda$.

We have $\PP(E) = \PP(E_2|E_1)\cdot \PP(E_1).$ Now, the fact that $N_{t/2}/e^{\lambda t/2}$ has a limiting distribution as $\lambda$ tends to infinity (an exponential distribution with a mean of $1$)
implies that $\PP(E_1) \geq \delta'>0$ for a fixed $\delta'>0$ (we can take for $\delta'$ any number smaller than $e^{-\frac{1}{2}} - e^{-\frac{3}{2}}$ for large enough values of $\lambda$).
Moreover, $\EE[N_t|E_1] \leq \frac{3}{2}n^2$, since
$\EE[N_t|N_{t/2} = k] = ke^{\lambda t/2} = kn  \leq \frac{3}{2}n^2$ for any $k \leq 3n/2$.  However:
$$\EE[N_t|E_1] \geq 2n^2\cdot \PP(N_t \geq 2n^2|E_1) = 2n^2(1-\PP(E_2|E_1)).$$  Thus, $\PP(E_2|E_1) \geq \frac{1}{4}$, and so,
$\PP(E) \geq \frac{1}{4}\delta' = :\delta>0$ as claimed.

 Suppose the number of individuals at time $t/2$ is $m$;  label them $1, 2, \ldots, m$.  For individual $i$, let $n_i$ be the number of descendants at time $t$.
 Thus $\sum_{i=1}^m n_i$ is the total number of individuals at time $t$. Now we use a well-known property of the (discrete) Yule distribution -- for a binary tree with $n_i$ leaves, the probability that the root is incident with a leaf is exactly $2/n_i$. Now individual $i \in \{1, \ldots, m\}$ is not the root of a binary tree, but if the binary tree below $i$ has the property just described, then either the edge $i$ lies on, or an edge in the binary tree below it, has a length of at least $t/4$.  Also if $n_i \leq 2$ then once again we must have at least one edge with a length of at least $t/4$.

For any particular value of $m$ that satisfies event $E_1$, let $p$ be the probability that none of the $m$ individuals gives rise in this way to an edge of length at least $t/4$. Then $p$ is bounded above  (by independence) as follows:
 \begin{equation}
 \label{pleq}
p \leq \prod_{i=1}^m (1-\frac{2}{n_i}),
\end{equation}
where the $n_i$ values satisfy constraints implied by $E$: $$\sum_{i=1}^m n_i \leq 2n^2, \mbox{ and } m \geq \frac{1}{2}n,$$
as well as our assumption $n_i \geq 2$ for all $i$.
Maximizing the term on the right-hand side of (\ref{pleq}) subject to the constraint $\sum_{i=1}^m n_i \leq 2n^2$,   we have:
$$p \leq (1-\frac{2m}{2n^2})^m \sim  e^{-m^2/n^2} \leq e^{-0.25}.$$
Thus, with probability at least $\delta(1 - e^{-0.25})$ there is an edge in the Yule tree having length  at least $t/4$. This completes the proof of Proposition
\ref{yuleprop}.

{\bf Proof of Lemma~\ref{covlem}}. Suppose $f'(y) \geq c>0$ for all $y \in [0,l]$ and that
$Y$ is discrete taking finite values $l \geq y_1 \geq y_2 \geq \cdots \geq y_n \geq 0$ (other cases are similar), and let
$p(y) = \PP(Y = y)$.  Then evaluating the following double sum by expanding out terms gives us the identity:
\begin{equation}
\label{sum1}
\sum_{i,j} (y_i-y_j)(f(y_i)-f(y_j))p(y_i)p(y_j) = 2{\rm Cov}[Y, f(Y)].
\end{equation}
However we can also write this double sum in the form:
\begin{equation}
\label{sum2}
2\sum_{i,j: i>j} (y_i-y_j)(f(y_i)-f(y_j))p(y_i)p(y_j) \geq 2c\sum_{i,j: i>j} (y_i-y_j)^2p(y_i)p(y_j),
\end{equation}
where the inequality holds since, for $y_i \geq y_j$ the condition $f'(y) \geq c$ for all $y \in [0,l]$ implies that
$f(y_i)-f(y_j) \geq c(y_i - y_j)$ by the mean value theorem.  Now,
$$2c\sum_{i,j: i>j} (y_i-y_j)^2p(y_i)p(y_j) = c\sum_{i,j} (y_i-y_j)^2p(y_i)p(y_j) = 2c{\rm Var}[Y].$$
Applying this to Eqns. (\ref{sum1}) and (\ref{sum2}) gives the result claimed.


\begin{thebibliography}{99}

\bibitem{aldous}
D. Aldous, Stochastic models and descriptive statistics for phylogenetic trees, from Yule to today, Stat. Sci. 16 (2001), 23--34.

\bibitem{ald}
D. Aldous, J.A Fill, Reversible Markov chains and random walks on graphs, Chapter 3, Eqn. (40), 2010,
http://www.stat.berkeley.edu/~aldous/RWG/book.html.


\bibitem{cha}
J. Chang, Full reconstruction of Markov models on evolutionary trees: identifiability and consistency, Math. Biosci. 137 (1996), 51--73.

\bibitem{des}
R. Desper, O. Gascuel, Theoretical foundation of the balanced minimum evolution method of phylogenetic inference and its relationship to weighted least-squares tree fitting, Mol. Biol. Evol. 21(3) (2004), 587--98.


\bibitem{eva}
W. Evans, C. Kenyon, Y. Peres, L.J. Schulman, Broadcasting on trees and the Ising Model, Ann. Appl. Probab. 10(2) (2000), 410--433.


\bibitem{kos}
J.M. Koshi, R.A. Goldstein, Models of natural mutations including site heterogeneity, Proteins, 32 (1998), 289--295.

\bibitem{lar}
N. Lartillot, H. Philippe, A Bayesian mixture model for across-site heterogeneities in the amino-acid replacement process,
Mol. Biol. Evol. 21 (2004), 1095--1109.

\bibitem{li}
G. Li, J. Ma, L. Zhang, Greedy selection of species for ancestral state reconstruction on phylogenies: Elimination is better than insertion, PLoS ONE 5(2) e8985 (2010).

\bibitem{lib}
D.A. Liberles, Ancestral sequence reconstruction, Oxford University Press, New York, 2007.

\bibitem{mck}
A. McKenzie, M.A. Steel, Properties of phylogenetic trees generated by Yule-type speciation models, Math. Biosci. 170 (2001), 91--112.

\bibitem{mad} W.P. Maddison, Calculating the probability distributions of ancestral states reconstructed by parsimony on phylogenetic trees, Syst. Biol. 44 (1995), 474--481.

\bibitem{mos}  E. Mossel, On the impossibility of reconstructing ancestral data and phylogenies,
J. Comput. Biol.  10(5) (2003), 669--678.

\bibitem{pag} M. Pagel, The maximum likelihood approach to reconstructing ancestral character states of discrete characters on phylogenies, Syst. Biol. 48 (1999), 612--622.

\bibitem{sal} B.A. Salisbury, J. Kim, Ancestral state estimation and taxon sampling density, Syst. Biol. 50(4) (2001), 557--564.


\bibitem{sob}  E. Sober, Temporally asymmetric inference in a Markov Process, Phil. Sci. 58(3) (1991), 398--410.

\bibitem{sob2} E. Sober, Evolution and Evidence: The logic behind the science, Cambridge University Press, Cambridge, UK, 2008.

\bibitem{ste}  M. Steel, D. Penny, Maximum parsimony and the phylogenetic information in multi-state characters, in: V. Albert (Ed.), Parsimony, phylogeny and genomics, Oxford University Press, 2005, pp. 163--178.

\bibitem{str} S.H. Strogatz, Nonlinear dynamics and chaos, Addison-Wesley Publishing Company, 1994.

\bibitem{swo}  D.L. Swofford, G.J. Olsen, P.J. Waddell, D.M. Hillis, Phylogenetic Inference, in: D.M. Hillis, C. Moritz, B.K. Mable (Eds.), Molecular Systematics, Second Edition, Sinauer Press, 1996.


\bibitem{tease}
M.A. Steel, L.A. Sz{\'e}kely, Teasing apart two trees. Comb. Probab. Comput. 16 (2007), 903--922.

\bibitem{tuf}
C. Tuffley, M.A. Steel, Links between maximum likelihood and maximum parsimony under a simple model of site substitution, B. Math. Biol. 59(3) (1997), 581--607.



\bibitem{zha} J. Zhang, M. Nei, Accuracies of ancestral amino acid sequences inferred by the parsimony, likelihood, and distance methods. J. Mol. Evol. 44 (1997), S139--S146.

\end{thebibliography}
\end{document}